\documentclass{llncs}

\usepackage{fancyhdr}
\usepackage{amssymb,amsmath}
\usepackage[usenames,dvipsnames]{color}
\usepackage[colorlinks=true,linkcolor=blue,citecolor=Mahogany]{hyperref}

\usepackage[comma,authoryear,square,sort]{natbib} 
\usepackage{multirow}
\usepackage{mathtools}
\usepackage{bbm}

\newcommand{\eps}{\varepsilon}
\renewcommand{\epsilon}{\eps}

\newcommand{\ignore}[1]{}

\newcommand{\YES}{\textsc{Yes}}

\newtheorem{reductionrule}{\bf Reduction rule}

\usepackage{xspace}

\newcommand{\Pb}{\ensuremath{\mathsf{P}}\xspace}
\newcommand{\NPCb}{\ensuremath{\mathsf{NPC}}\xspace}
\newcommand{\caveat}{\ensuremath{\mathsf{CoNP\subseteq NP/Poly}}\xspace}

\newcommand{\CC}{{\mathcal C}}

\newcommand{\FF}{{\mathcal F}}

\newcommand{\HH}{{\mathcal H}}

\newcommand{\UU}{{\mathcal U}}
\newcommand{\VV}{{\mathcal V}}

\newcommand{\defparproblem}[4]{
\begin{center}
\noindent\fbox{

  \begin{minipage}{\textwidth}
  \begin{tabular*}{\textwidth}{@{\extracolsep{\fill}}lr} \textsc{#1}  & {\bf{Parameter:}} #3 \\ \end{tabular*}
  {\bf{Input:}} #2  \\
  {\bf{Question:}} #4
  \end{minipage}
 
  }
\end{center}
}

\begin{document}

\title{Kernelization Complexity of Possible Winner and Coalitional Manipulation Problems in Voting}

\author{Palash Dey, Neeldhara Misra, and Y. Narahari\\ \texttt{palash@csa.iisc.ernet.in, mail@neeldhara.com, hari@csa.iisc.ernet.in}}

\institute{Department of Computer Science and Automation \\Indian Institute of Science - Bangalore, India.\\[10pt]Date: \today}

\maketitle
\pagestyle{fancy}


\begin{abstract}
In the \textsc{Possible Winner} problem in computational social
choice theory, we are given a set of partial preferences and the question
is whether a distinguished candidate could be made winner by extending the
partial preferences to linear preferences. Previous work has provided, for many
common voting rules, fixed parameter tractable algorithms 
for the \textsc{Possible Winner} problem, with
number of candidates as the parameter. 
However, the corresponding kernelization question is still
open and in fact, has been mentioned as a key research challenge~\citep{ninechallenges}. 
In this paper, we settle this open question for many common voting rules. 

We show that the \textsc{Possible Winner} problem for
maximin, Copeland, Bucklin, ranked pairs, and a class of scoring rules that
include the Borda voting rule do not admit a polynomial kernel with the number
of candidates as the parameter. We show however that the \textsc{Coalitional Manipulation}
problem which is an important special case of the \textsc{Possible Winner} problem
does admit a polynomial kernel for maximin, Copeland, ranked pairs, and a class of
scoring rules that includes the Borda voting rule, when the number of
manipulators is polynomial in the number of candidates. 
A significant conclusion of our work is that the \textsc{Possible Winner} problem 
is harder than the \textsc{Coalitional Manipulation} problem since the \textsc{Coalitional Manipulation} 
problem admits a polynomial kernel whereas the \textsc{Possible Winner} problem does not admit 
a polynomial kernel. 
\end{abstract}

\keywords{Computational social choice, possible winner, voting, kernelization, parameterized complexity}

\newpage

\section{Introduction}
\noindent In many real life situations including multiagent systems, agents often need to aggregate 
their preferences and agree upon a common decision (candidate). 
Voting is an immediate natural tool in these 
situations. Common and classical applications of voting rules in artificial intelligence 
include collaborative filtering \citep{pennock2000social}, 
planning among multiple automated agents \citep{ephrati1991clarke}, etc.

Usually, in a voting setting, it is assumed that the votes are complete orders over the
candidates. However, due to many reasons, for example, lack of knowledge of voters 
about some candidates, a voter may be indifferent between some pairs of candidates. 
Hence, it is both natural and important to 
consider scenarios where votes are partial orders over the candidates. When votes are only partial orders over the candidates, the winner cannot be determined with certainty since it depends on 
how these partial orders are extended to linear orders. 
This leads to a natural computational problem called the \textsc{Possible Winner}~\citep{konczak2005voting} problem:  given a set of partial votes $P$ and a distinguished candidate $c$, is there a way to extend the partial votes to linear ones to make $c$ win?
The \textsc{Possible Winner} problem has been studied extensively in the literature 
\citep{lang2007winner,pini2007incompleteness,walsh2007uncertainty,xia2008determining,betzler2009multivariate,betzler2009towards,chevaleyre2010possible,betzler2010partial,baumeister2011computational,lang2012winner,faliszewski2014complexity} following its definition in \citep{konczak2005voting}.
The \textsc{Possible Winner} problem is known to be in $\NPCb$ for many common voting rules, for example, scoring rules, maximin, 
Copeland, Bucklin, and ranked pairs etc.~\citep{xia2008determining}. Walsh \citep{walsh2007uncertainty} showed, for a constant number of candidates, that the \textsc{Possible Winner} problem can be solved in polynomial time for all
the voting rules mentioned above. 
An important special case of the \textsc{Possible Winner} problem is the 
\textsc{Coalitional Manipulation} problem \citep{bartholdi1989computational} 
where only two kinds of partial votes 
are allowed - complete preference and empty preference. The set of empty votes 
is called the manipulators' vote and is denoted by $M$. The \textsc{Coalitional Manipulation} problem 
is in $\NPCb$ for maximin, Copeland, and ranked pairs voting rules even when $|M|\ge 2$ 
\citep{faliszewski2008copeland,faliszewski2010manipulation,xia2009complexity}. 
The \textsc{Coalitional Manipulation} problem is in $\Pb$ for the Bucklin voting rule 
\citep{xia2009complexity}. We refer to \citep{xia2008determining,walsh2007uncertainty,xia2009complexity} 
for detailed overviews.

\subsection{Our Methodology} Preprocessing, as a strategy for coping with hard problems, is 
universally applied in practice. The main goal here is \emph{instance compression} - the objective is to output a smaller instance while maintaining equivalence. In the classical setting, $\mathbb{NP}$-hard problems are unlikely to have efficient compression algorithms (since repeated application would lead to an efficient solution for the entire problem, which is unexpected). However, the breakthrough notion of \emph{kernelization} in parameterized complexity provides a mathematical framework for analyzing the quality of preprocessing strategies. In parameterized 
complexity, each problem instance $(x,k)$ comes 
with a parameter $k$. 
The parameterized problem is said to admit a \emph{kernel} if there is a 
polynomial time algorithm (where the degree of polynomial is independent of $k$), 
called a \emph{kernelization algorithm}, 
that reduces the input instance to an instance with size bounded by 
a function of $k$, while preserving the answer. 
This has turned out to be an important and widely applied notion in theory, and has also proven very successful in practice~\citep{weihe1998covering,lokshtanov2012kernelization}. Quantitatively, running a kernelization algorithm before 
solving it using an algorithm that runs in time $f(|x|)$ brings down the running time to $f(k)+p(|x|)$, where $|x|$ is the 
size of the input instance and the running time of the kernelization algorithm is $p(|x|)$.

A problem with parameter $k$ is called \emph{fixed parameter tractable} (FPT) if it is solvable in 
time $f(k) \cdot p(|x|)$, where $f$ is an arbitrary function of $k$ and 
$p$ is a polynomial in the input size $|x|$. 
The existence of a fixed parameter tractable algorithm implies existence of a kernel for that problem. However, 
the size of the kernel need not be polynomial in the parameter. 
A polynomial kernel is said to exist if there 
is a kernelization algorithm that can output an equivalent problem instance of size polynomial in the 
parameter. We refer to \citep{downey1999parameterized,niedermeier2006invitation} 
for an excellent overview on  fixed parameter algorithms and kernelization. 

\subsection{Contributions} Discovering kernelization algorithms is currently an active and interesting area of 
research in computational social choice theory~\citep{betzler2010partial,betzler2010problem,betzler2009fixed,bredereck2012multivariate,froese2013parameterized,bredereck2014prices,Betzlerxyz,dorn2012multivariate}. 
Betzler et al. \citep{betzler2009multivariate} showed that the 
\textsc{Possible Winner} problem has a kernelization algorithm when parameterized 
by the total number of candidates for scoring rules, maximin, Copeland, Bucklin, and ranked 
pairs voting rules. A natural and practical follow-up question is whether the 
problem admits a polynomial kernel when parameterized by the number of candidates. This question has been open ever since
the work of Betzler et al. and in fact, has been mentioned as a key research challenge in parameterized algorithms for computational social choice theory~\citep{ninechallenges}. Betzler et al. showed non-existence of polynomial kernel for the \textsc{Possible Winner} problem for the $k$-approval voting rule when parameterized by $(t,k)$, where $t$ is the number of partial votes~\citep{betzler2010problem}. 
The $\NPCb$ reductions for the \textsc{Possible Winner} 
problem for scoring rules, maximin, Copeland, Bucklin, and ranked 
pairs voting rules given by Xia et al.~\citep{xia2009complexity} are from the \textsc{Exact 3 Set-Cover} 
problem. Their results do not throw any light on the existence of a polynomial kernel since \textsc{Exact 3 Set-Cover} has a trivial $O(m^3)$ kernel where $m$ is the size of the universe. In our work in this paper, 
we show that there is no polynomial kernel for the 
\textsc{Possible Winner} problem for maximin, Copeland, Bucklin, ranked pairs, and a class of scoring rules that include the Borda voting rule, 
when parameterized by the total number of candidates unless~\caveat{}. These hardness 
results are shown by a parameter-preserving many-to-one reduction from the 
\textsc{Small Universe Set Cover} problem for which there does not exist any polynomial kernel 
parameterized by universe size unless~\caveat{} \citep{DomLokSau2009}. 

On the other hand, we show that 
the \textsc{Coalitional Manipulation} problem admits a polynomial kernel for maximin, 
Copeland, ranked pairs, and a class of scoring rules that includes the Borda voting rule when we have $O(poly(m))$ number of manipulators -- 
specifically, we exhibit  an $O(m^2|M|)$ kernel for maximin and Copeland, and an $O(m^4|M|)$ kernel for the ranked pairs voting rule, where $m$ is the number of candidates and $M$ is the set of manipulators. 
The \textsc{Coalitional Manipulation} problem for the Bucklin voting rule is 
in $\Pb$ \citep{xia2009complexity} and thus the kernelization question does not arise. 

A significant conclusion of our work is that, although the \textsc{Possible Winner} and 
\textsc{Coalitional Manipulation} problems are both $\NPCb$, the \textsc{Possible Winner} problem 
is harder than the \textsc{Coalitional Manipulation} problem since the \textsc{Coalitional Manipulation} 
problem admits a polynomial kernel whereas the \textsc{Possible Winner} problem does not admit 
a polynomial kernel. 

The rest of the paper is organized as follows. We first establish the setup and general notions. 
Next, we discuss non-existence of  polynomial kernels for the \textsc{Possible Winner} problem followed by 
existence of polynomial kernels for the \textsc{Coalitional Manipulation} problem.

This paper is a significant extension of the conference version of this work~\cite{deykernel}: this extended version includes all the proofs.

\section{Preliminaries}

Let $\mathcal{V}=\{v_1, \dots, v_n\}$ be the set of all 
\emph{voters} and $\mathcal{C}=\{c_1, \dots, c_m\}$ 
the set of all \emph{candidates}. From, here on, $n$ is the number of voters and $m$ is the 
number of candidates unless mentioned otherwise. 
Each voter $v_i$'s \textit{vote} is a \emph{preference} $\succ_i$ over the 
candidates which is a linear order over $\mathcal{C}$. 
For example, for two candidates $a$ and $b$, $a \succ_i b$ means that the voter $v_i$ prefers $a$ to $b$. 
We will use $a >_i b$ to denote the fact that $ a \succ_i b, a \ne b$. 
We denote the set of all linear orders over $\mathcal{C}$ by $\mathcal{L(C)}$. 
The set of all preference profiles $(\succ_1, \dots, \succ_n)$ of 
the $n$ voters is denoted by $\mathcal{L(C)}^n$. 
We use $x\% y$ to denote the modulus operation $x \text{ mod } y$ (for example, $5\%2$ is the same as $1$). 
We denote the set $\{1, 2, \dots\}$ by $\mathbb{N}^+$ and  the set $\{1, \cdots, k\}$ by $[k]$. 
Let $\uplus$ denote the disjoint union of sets. 
A map $r:\uplus_{n,|\mathcal{C}|\in\mathbb{N}^+}\mathcal{L(C)}^n\longrightarrow 2^\mathcal{C}\setminus\{\emptyset\}$
is called a \emph{voting rule}. For a voting rule $r$ and a preference profile $\succ$, we say a candidate $x$ wins uniquely if $r(\succ) = \{x\}$. 
In this paper, unless mentioned otherwise, winning means winning uniquely. 
A more general setting is an {\em election\/} where the votes are only 
\emph{partial orders} over candidates. A \emph{partial order} is a relation that is \emph{reflexive, 
antisymmetric}, and \emph{transitive}. A partial vote can be extended to possibly more than one linear votes depending on how 
we fix the order for the unspecified pairs of candidates. For example, in an election with the set of candidates $\mathcal{C} = \{a, b, c\}$, 
a valid partial vote can be $a \succ b$. This partial vote can be extended to three linear votes namely, $a \succ b \succ c$, $a \succ c \succ b$, 
$c \succ a \succ b$. However, the voting rules always take as input a set of votes 
that are complete orders over the candidates. 
Given an election $E$, we can construct a weighted graph $G_E$ called \textit{weighted majority graph} from $E$. The set of vertices in $G_E$ is the set of candidates in $E$. For any two candidates $x$ and $y$, the weight on the edge $(x,y)$ is $D_E(x,y) = N_E(x,y) - N_E(y,x)$, where $N_E(x,y)(N_E(y,x))$ is the number of voters who prefer $x$ to $y$ ($y$ to $x$). The following voting rules use weighted majority graph to select winner.
Some examples of common voting rules are as follows.

\begin{itemize}
 \item {\bf Positional scoring rules:} Given an $m$-dimensional vector $\vec{\alpha}=\left(\alpha_1,\alpha_2,\dots,\alpha_m\right)\in\mathbb{R}^m$ 
 with $\alpha_1\ge\alpha_2\ge\dots\ge\alpha_m$, we can naturally define a
 voting rule - a candidate gets score $\alpha_i$ from a vote if it is placed at the $i^{th}$ position, and the 
 score of a candidate is the sum of the scores it receives from all the votes. 
 The winners are the candidates with maximum score. Scoring rules remain unchanged if we multiply every $\alpha_i$ by any constant $\lambda>0$ and/or add any constant $\mu$. Hence, we assume without loss of generality that for any score vector $\vec{\alpha}$, there exists a $j$ such that $\alpha_j - \alpha_{j+1}=1$ and $\alpha_k = 0$ for all $k>j$. We call such an $\vec{\alpha}$ a normalized score vector. A scoring rule is called a 
 strict scoring rule if $\alpha_1>\alpha_2>\dots>\alpha_m$.
 For $\alpha=\left(m-1,m-2,\dots,1,0\right)$, we get the \emph{Borda} voting rule. With $\alpha_i=1$ $\forall i\le k$ and $0$ else, 
 the voting rule we get is known as $k$-\emph{approval}. \emph{Plurality} is $1$-\emph{approval} and \emph{veto} is $(m-1)$-\emph{approval}.
 
 \item {\bf Bucklin:} A candidate $x$'s Bucklin score is the minimum number $l$ such that more than half 
 of the voters rank $x$ in their top $l$ positions. The winners are the candidates with lowest Bucklin score.
 
 \item {\bf Maximin:} The maximin score of a candidate $x$ is $\min_{y\ne x} D(x,y)$. The winners are the candidates with maximum maximin score.
 
 \item {\bf Copeland:} The Copeland score of a candidate $x$ is the number of candidates $y\ne x$ such that $D(x,y)>0$. The winners are the candidates with maximum Copeland score.
 
 \item {\bf Ranked pairs:} We pick the pair $(c_i, c_j) \in \CC \times \CC$ such that $D(c_i,c_j)$ is maximum. 
 We fix the ordering between $c_i$ and $c_j$ to be 
 $c_i \succ c_j$ unless it contradicts previously fixed orders. We continue this process 
 until all pairwise elections are considered. At this point, we have a complete order over the candidates. 
 Now, the top candidate is chosen as the winner.
\end{itemize}

We use the parallel-universes tie breaking~\citep{conitzer2009preference,brill2012price} 
to define the winning candidate for
the ranked pairs voting rule. In this setting, 
a candidate $c$ is a winner if and only if there exists a way to 
break ties in all of the steps such that $c$ is the winner. 
We say that a candidate is the unique winner of a ranked pairs election if the candidate is the 
sole winner in all possible tie breakings in all the steps. 


We now briefly describe the framework in which we analyze the computational complexity of the 
\textsc{Possible Winner} problem.


\paragraph*{Parameterized Complexity.}~A parameterized problem $\Pi$ is a 
subset of $\Gamma^{*}\times
\mathbb{N}$, where $\Gamma$ is a finite alphabet. An instance of a
parameterized problem is a tuple $(x,k)$, where $k$ is the
parameter. A {\em kernelization} algorithm is a set of preprocessing rules that runs in 
polynomial time and reduces the instance size with a guarantee on the output instance size. This notion is 
formalized below.
\begin{definition}{\rm \bf[Kernelization]}~\citep{niedermeier2006invitation,flum2006parameterized}
A kernelization algorithm for a parameterized problem   $\Pi\subseteq \Gamma^{*}\times \mathbb{N}$ is an 
algorithm that, given $(x,k)\in \Gamma^{*}\times \mathbb{N} $, outputs, in time polynomial in $|x|+k$, a pair 
$(x',k')\in \Gamma^{*}\times
  \mathbb{N}$ such that (a) $(x,k)\in \Pi$ if and only if
  $(x',k')\in \Pi$ and (b) $|x'|,k'\leq g(k)$, where $g$ is some
  computable function.  The output instance $x'$ is called the
  kernel, and the function $g$ is referred to as the size of the
  kernel. If $g(k)=k^{O(1)}$ then we say that
  $\Pi$ admits a polynomial kernel.
\end{definition}

For many parameterized problems, it is well established that the existence of a polynomial kernel would 
imply the collapse of the polynomial hierarchy to the third level (or more precisely, \caveat{}). 
Therefore, it is considered unlikely that these problems would admit polynomial-sized kernels. 
For showing kernel lower bounds, we simply establish reductions from these problems. 

\begin{definition}{\rm \bf[Polynomial Parameter Transformation]}
\label{def:ppt-reduction} {\rm\citep{BodlaenderThomasseYeo2009}}
Let $\Gamma_1$ and $\Gamma_2$ be parameterized problems. We say that $\Gamma_1$ is
polynomial time and parameter reducible to $\Gamma_2$, written
$\Gamma_1\le_{Ptp} \Gamma_2$, if there exists a polynomial time computable
function $f:\Sigma^{*}\times\mathbb{N}\to\Sigma^{*}\times\mathbb{N}$, and a
polynomial $p:\mathbb{N}\to\mathbb{N}$, and for all
$x\in\Sigma^{*}$ and $k\in\mathbb{N}$, if
$f\left(\left(x,k\right)\right)=\left(x',k'\right)$, then
$\left(x,k\right)\in \Gamma_1$ if and only if $\left(x',k'\right)\in \Gamma_2$,
and $k'\le p\left(k\right)$. We call $f$ a polynomial parameter
transformation (or a PPT) from $\Gamma_1$ to $\Gamma_2$.
\end{definition}

This notion of a reduction is useful in showing kernel lower bounds
because of the following theorem.

\begin{theorem}\label{thm:ppt-reduction}~{\rm\cite[Theorem 3]{BodlaenderThomasseYeo2009}}
  Let $P$ and $Q$ be parameterized problems whose derived
  classical problems are $P^{c},Q^{c}$, respectively. Let $P^{c}$
  be $\NPCb{}$, and $Q^{c}\in$ $\mathbb{NP}$. Suppose there exists a PPT from
  $P$ to $Q$.  Then, if $Q$ has a polynomial kernel, then $P$ also
  has a polynomial kernel. \end{theorem}
  
\section{Hardness of Kernelization for Possible Winner Problem}

In this section, we show non-existence of polynomial kernels for the \textsc{Possible 
Winner} problem for maximin, Copeland, Bucklin, ranked Pairs, and a class of scoring rules that includes the Borda voting rule. 
We do this by demonstrating polynomial parameter transformations from the \textsc{Small Universe Set Cover} problem, which 
is the classic \textsc{Set Cover} problem, but now parameterized by the size of the universe and the budget.
\defparproblem{Small Universe Set Cover}
{
A set $\UU = \{u_1,\ldots,u_m\}$ and a family $\FF = \{S_1, \ldots, S_t\}$. 
}
{$m+k$}
{
Is there a subfamily $\HH \subseteq \FF$ of size at most $k$ such that every element of the universe belongs to 
at least one $H \in \HH$?}

It is well-known~\citep{DomLokSau2009} that \textsc{Red-Blue Dominating Set} parameterized by $k$ and the number 
of non-terminals does not admit a polynomial kernel unless~\caveat{}. It follows, by the  
duality between dominating set and set cover, that \textsc{Set Cover} when parameterized by the solution size 
and the size of the universe (in other words, the \textsc{Small Universe Set Cover} problem defined above) does 
not admit a polynomial kernel unless~\caveat{}. 

We now consider the \textsc{Possible Winner} problem parameterized by the number of candidates for maximin, 
Copeland, Bucklin, ranked pairs, and a class of scoring rules that includes the Borda rule,
and establish that they do not admit a polynomial kernel 
unless~\caveat{}, by polynomial parameter transformations from \textsc{Small Universe Set Cover}. 

\subsection{Scoring Rules}

First, we prove hardness of kernelization for the \textsc{Possible Winner} problem for a large class of scoring 
rules that includes the Borda rule. For that, we use the following lemma which has been used before~\citep{baumeister2011computational}.

\begin{lemma}\label{score_gen}
Let $\mathcal{C} = \{c_1, \ldots, c_m\} \uplus D, (|D|>0)$ be a set of candidates, and $\vec{\alpha}$ a normalized score vector of length $|\mathcal{C}|$. Then, for any given $\mathbf{X} = (X_1, \ldots, X_m) \in \mathbb{Z}^m$, there exists $\lambda\in \mathbb{R}$ and a voting profile such that the $\vec{\alpha}$-score of $c_i$ is $\lambda + X_i$ for all $1\le i\le m$,  and the score of candidates $d\in D$ is less than $\lambda$. Moreover, the number of votes is $O(poly(|\mathcal{C}|\cdot \sum_{i=1}^m |X_i|))$.
\end{lemma}

With the above lemma at hand, we now show hardness of polynomial kernel result for the class of strict scoring rules.

\begin{theorem}
\label{thm:npk_borda}
The \textsc{Possible Winner} problem for any strict scoring rule, when parameterized by the number of candidates, does not admit a polynomial kernel unless~\caveat{}.
\end{theorem}

\begin{proof}
 Let $(\UU,\FF,k)$ be an instance of \textsc{Small Universe Set Cover}, where $\UU = \{u_1,\ldots,u_m\}$ and $\FF = \{S_1, \ldots, S_t\}$. We use  $T_i$ to denote $\UU \setminus S_i$. We let $(\alpha_1, \alpha_2, \ldots, \alpha_t)$ denote the score vector of length $t$, and let $\delta_i$ denote the difference $(\alpha_{i} - \alpha_{i+1})$. Note that for a strict scoring rule, all the $\delta_i$'s will be strictly positive. We now construct an instance $(\CC,V,c)$ of \textsc{Possible Winner} as follows.  
 
\paragraph*{Candidates.} $\CC = \UU \uplus \VV \uplus \{w,c,d\}$, where $\VV := \{v_1,...,v_m\}$. 

\paragraph*{Partial Votes, $P$.}~The first part of the voting profile comprises $t$ partial votes, and will be denoted by $P_1$. Let $V_j$ denote the set $\{v_1, \ldots, v_j\}$. For each $1 \leq i \leq t$, we first consider a profile built on a total order $\eta_i$:
$$\eta_i := d \succ S_i \succ V_j \succ w \succ \cdots, \mbox{where } j = m - |S_i|.$$ 
Now, we obtain a partial order $\lambda_i$ based on $\eta_i$ as follows:
$$\lambda_i := \eta_i \setminus \left( \{w\} \times \left( \{d\} \uplus S_i \uplus V_j\right)\right)$$

Let $P_1^\prime$ be the set of votes $\{\eta_i ~|~ 1 \leq i \leq t\}$.  We will now add more votes, which we denote by $P_2$, such that 
$s(c) = s(u_i), s(d) - s(c) = (k -1)\delta_1, s(c) - s(w) = k(\delta_2 + \delta_3 + \cdots + \delta_{m+1}) + \delta_1, s(c) > s(v_i) + 1$, for all $1 \leq i \leq t$ where $s(\cdot)$ is the scores of the candidates from the combined voting profile $P_1^\prime \uplus P_2$. From the proof of Lemma\nobreakspace \ref {score_gen}, we can see that such votes can always be constructed. In particular, also note that the voting profile $P_2$ consists of complete votes. Note that the number of candidates is $2m+3$, which is polynomial in the size of the universe, as desired. We now claim that the candidate $c$ is a possible winner for the voting profile $P_1 \uplus P_2$ with respect to the strict scoring rule if, and only if, $(\UU,\FF)$ is a YES-instance of Set Cover. 

In the forward direction, suppose, without loss of generality, that $S_1, \ldots, S_k$ form a set cover. Then we propose the following extension for the partial votes $\lambda_1, \ldots, \lambda_k$:
$$w > d > S_i > V_j > ... ,$$

and the following extension for the partial votes $\lambda_{k+1}, \ldots, \lambda_t$:
$$d > S_i > V_j > w > ... $$

For $1 \leq i \leq k$, the position of $d$ in the extension of $\lambda_i$ proposed above is $\delta_1$ lower than its original position in $\eta_i$. Therefore, the final score of $d$ is lower than the score of $c$. Similarly, the scores of all the $u_i$'s decrease by at least $\min_{i=2}^m \{\delta_i\}$, which is strictly positive, since the $S_i$'s together form a set cover, and the scoring rule is strict. Finally, the score of $w$ increase by at most $k(\delta_2 + \delta_3 + \cdots + \delta_{m+1})$, since there are at most $k$ votes where the position of $w$ in the extension of $\lambda_i$ improved from it's original position in $\eta_i$. Therefore, the score of $c$ is greater than any other candidate, implying that $c$ is a possible winner. 

For the reverse direction, notice that there must be at least $k$ extensions where $d$ is in the second position, since the score of $d$ is $(k-1)\delta_1$ more than the score of $c$. In these extensions, observe that $w$ will be at the first position. On the other hand, placing $w$ in the first position causes its score to increase by $(\delta_2 + \delta_3 + \cdots + \delta_{m+1})$, therefore, if $w$ is in first position in $\ell$ extensions, its score increases by $\ell(\delta_2 + \delta_3 + \cdots + \delta_{m+1})$. Since the score difference between $w$ and $c$ is only $k(\delta_2 + \delta_3 + \cdots + \delta_{m+1}) + 1$, we can afford to have $w$ in the first position in \emph{at most} $k$ votes. Therefore, apart from the extensions where $d$ is in the second position, in all remaining extensions, $w$ appears after $V_j$, and therefore the candidates from $S_i$ continue to be in their original positions. We now claim that the sets corresponding to the $k$ votes where $d$ is on top form a set cover.  Indeed, if not, suppose the element $u_i$ is not covered. It is easily checked that the score of such a $u_i$ remains unchanged in this extension, and therefore its score is equal to $c$, contradicting our assumption that we started with an extension for which $c$ was a winner. 
\qed\end{proof}

The proof of Theorem\nobreakspace \ref {thm:npk_borda} can be generalized to to a wider class of scoring rules as stated in the following corollary.
\begin{corollary}
 Let $r$ be a positional scoring rule such that there exists a polynomial function $f:\mathbb{N}\rightarrow \mathbb{N}$, such that for every $m\in \mathbb{N}$, there exists an index $l$ in the $f(m)$ length score vector $\vec{\alpha}$ satisfying following,
 $$ \alpha_i - \alpha_{i+1} > 0, \forall l\le i\le l+m $$
 Then, the \textsc{Possible Winner} problem for $r$, when parameterized by the number of candidates, does not admit a polynomial kernel unless~\caveat{}.
\end{corollary}

\subsection{Maximin Voting Rule}

We will need the following lemma in subsequent proofs. The lemma has been used before \citep{mcgarvey1953theorem,xia2008determining}.

\begin{lemma}\label{thm:mcgarvey}
 For any function $f:\mathcal{C} \times \mathcal{C} \longrightarrow \mathbb{Z}$, such that
 \begin{enumerate}
  \item $\forall a,b \in \mathcal{C}, f(a,b) = -f(b,a)$.
  \item either $\forall a,b \in \mathcal{C}, f(a,b)$ is even or 
  $\forall a,b \in \mathcal{C}, f(a,b)$ is odd.
 \end{enumerate}
 there exists a $n$ voters' profile such that for all $a,b \in \mathcal{C}$, $a$ defeats 
 $b$ with a margin of $f(a,b)$. Moreover, 
 \[n = O\left(\sum_{\{a,b\}\in \mathcal{C}\times\mathcal{C}} |f(a,b)|\right)\]
\end{lemma}

We now describe the reduction for the \textsc{Possible Winner} parameterized by the number of candidates, for the maximin voting rule.

\begin{theorem}
\label{thm:npk_maximin}
The \textsc{Possible Winner} problem for maximin voting rule, when parameterized by the number of candidates, does not admit a polynomial kernel unless~\caveat{}.
\end{theorem}

\begin{proof}
 Let $(\UU,\FF,k)$ be an instance of \textsc{Small Universe Set Cover}, where $\UU = \{u_1,\ldots,u_m\}$ and $\FF = \{S_1, \ldots, S_t\}$. We use  $T_i$ to denote $\UU \setminus S_i$. We now construct an instance $(\CC,V,c)$ of the \textsc{Possible Winner} as follows. 

\paragraph*{Candidates.}~$C := \UU \uplus W \uplus \{c,d,x\} \uplus L$. Where $W := \{w_1, w_2, \ldots, w_m, w_x\}, L:= \{l_1,l_2,l_3\}$.

\paragraph*{Partial Votes, $P$.}~The first part of the voting profile comprises $t$ partial votes, and will be denoted by $P$. For each $1 \leq i \leq t$, we first consider a profile built on a total order $\eta_i$. We denote the order $w_1 \succ \dots \succ w_m \succ w_x$ 
by $\vec{W}$. From this point onwards, whenever we place a set of candidates in some position of a partial order, we mean that the candidates in the set can be ordered arbitrarily. For example, the candidates in $S_i$ can be ordered arbitrarily among themselves in the total order $\eta_i$ below.
\[\eta_i := L \succ \vec{W} \succ x \succ S_i \succ d \succ c \succ T_i\]
Now, we obtain a partial order $\lambda_i$ based on $\eta_i$ as follows:
\[\lambda_i := \eta_i \setminus \left(W \times \left(\{c,d,x\} \uplus \UU\right)\right) \]
The profile $P$ consists of $\{ \lambda_i ~|~ 1 \leq i \leq t \}$.

\paragraph*{Complete Votes, $Q$.}~We now describe the remaining votes in the profile, which are linear orders designed to achieve specific pairwise difference scores among the candidates. This profile, denoted by $Q$, is defined according to~Lemma\nobreakspace \ref {thm:mcgarvey} to contain votes such that the pairwise score differences of $P \cup Q$ satisfy the following.

\begin{itemize}
\item $D(c,w_1) = -2k$.
\item $D(c,l_1) = -t$.
\item $D(d,w_1) = -2k - 2$.
\item $D(x,w_x) = -2k - 2$.
\item $D(w_i,u_i) = -2t$ $\forall$ $1 \leq i \leq m$.
\item $D(a_i,l_1) = D(w_x,l_1) = -4t$.
\item $D(l_1,l_2) = D(l_2,l_3) = D(l_3,l_1) = -4t$.
\item $D(l,r) \leq 1$ for all other pairs $(l,r) \in C \times C$.
\end{itemize}
We note that the for all $c,c' \in \CC$, the difference $|D(c,c') - D_P(c,c')|$ is always even, as long as $t$ is even and the number of sets in $\FF$ that contain an element $u \in \UU$ is always even. Note that the latter can always be ensured without loss of generality: indeed, if $u \in \UU$ occurs in an odd number of sets, then we can always add the set $\{u\}$ if it is missing and remove it if it is present, flipping the parity in the process. In case $\{u\}$ is the only set containing the element $u$, then we remove the set from both $\FF$ and $\UU$ and decrease $k$ by one. The number of sets $t$ can be assumed to be even by adding a dummy element and a dummy pair of sets that contains the said element. It is easy to see that these modifications always preserve the instance. 
Thus, the constructed instance of \textsc{Possible Winner} is $(\CC,V,c)$, where $V := P \cup Q$. We now turn to the proof of correctness. 

In the forward direction, let $\HH \subseteq \FF$ be a set cover of size at most $k$. Without loss of generality, let $|\HH| = k$ (since a smaller set cover can always be extended artificially) and let $\HH = \{S_1, \ldots, S_k\}$ (by renaming). 

If $i \leq k$, let: 
\[\lambda_i^* := L \succ x \succ S_i \succ d \succ c \succ \vec{W} \succ T_i\]
If $i > k$, let: 
\[\lambda_i^* := L \succ \vec{W} \succ x \succ S_i \succ d \succ c \succ T_i\]
Clearly $\lambda_i^*$ extends $\lambda_i$ for all $i \in 1 \leq i \leq k$. Let $V^*$ denote the extended profile consisting of the votes $\{ \lambda_i^* ~|~ 1 \leq i \leq k\} \cup Q$. We now claim that $c$ is the unique winner with respect to maximin voting rule in $V^*$. 

Since there are $k$ votes in $V^*$ where $c$ is preferred over $w_1$ and $(t-k)$ votes where $w_1$ is preferred to $c$, we have:
\begin{eqnarray*}
D_{V^*}(c,w_1) &=& D_{V}(c,w_1) + k - (t-k)
\\ &=& -2k + k - (t-k) = -t
\end{eqnarray*}
It is easy to check that maximin score of $c$ is $-t$. Also, it is straightforward to verify the following 
table.

\begin{table}[here]
\begin{minipage}{\textwidth}
  \begin{center}
  {\renewcommand{\arraystretch}{2}
 \begin{tabular}{|c|c| }\hline
  Candidate	& maximin score	\\\hline
  $w_i, \forall i\in \{1, 2, \dots, m\}$		& $ < -t $	\\\hline
  $u_i, \forall i\in \{1, 2, \dots, m\}$		& $\leq -4t$	\\\hline 
  $w_x$		& $\leq -4t$	\\\hline 
  $l_1, l_2, l_3$		& $\leq -4t$	\\\hline
  $x$		& $\leq -t-2$	\\\hline 
  $d$		& $\leq -t-2$\\\hline
 \end{tabular}
 }
 \end{center}
\end{minipage}
\label{summary}
\end{table}

Therefore, $c$ is the unique winner for the profile $V^*$.

We now turn to the reverse direction. Let $P^*$ be an extension of $P$ such that $V^* := P^* \cup Q$ admits $c$ as a unique winner with respect to maximin voting rule. We first argue that $P^*$ must admit a certain structure, which will lead us to an almost self-evident set cover for $\UU$.

Let us denote by $P^*_C$ the set of votes in $P^*$ which are consistent with $c \succ w_1$, and let $P^*_W$ be the set of votes in $P^*$ which are consistent with $w_1 \succ c$. We first argue that $P^*_C$ has at most $k$ votes.
\begin{claim} Let $P^*_C$ be as defined above. Then, $|P^*_C| \leq k$.
\end{claim}

\begin{proof}
Suppose, for the sake of contradiction, that more than $k$ extensions are consistent with $c \succ w_1$. Then we have:
\begin{eqnarray*}
D_{V^*}(c,w_1) & \geq & D_{V}(c,w_1) + k+1 - (t-k-1)
\\ &=&-2k + 2k - t + 2 = -t+2
\end{eqnarray*}
Since $D_{V^*}(c,l_1) = -t$, the maximin score of $c$ is $-t$. On the other hand, we also have that the maximin score of $d$ is given by $D_{V^*}(d,w_1)$, which is now at least $(-t)$:
\begin{eqnarray*}
D_{V^*}(d,w_1) & \geq & D_{V}(d,w_1) + k+1 - (t-k-1)
\\ &=&-2k-2 + 2k - t + 2 = -t
\end{eqnarray*}
Therefore, $c$ is no longer the unique winner in $V^*$ with respect to the maximin voting rule, 
which is the desired contradiction. 
\end{proof}

Next, we propose that a vote that is consistent with $w_1 \succ c$ must be consistent with $w_x \succ x$.
\begin{claim} 
Let $P^*_W$ be as defined above. Then, any vote in $P^*_W$ must respect $w_x \succ x$.
\end{claim}

\begin{proof}
Suppose there are $r$ votes in $P^*_C$, and suppose that in at least one vote in $P^*_W$ where $x \succ w_x$. Notice that any vote in $P^*_C$ is consistent with $x \succ w_x$. Now we have:
\begin{eqnarray*}
D_{V^*}(c,w_1) & = & D_{V}(c,w_1) + r - (t-r)
\\ &=&-2k + 2r - t 
\\ &=&-t - 2(k-r)
\end{eqnarray*}
And further:
\begin{eqnarray*}
D_{V^*}(x,w_x) & \geq & D_{V}(x,w_x) + (r+1) - (t-r-1)
\\ &=&-2k-2 + 2r - t + 2 
\\ &=&-t - 2(k-r)
\end{eqnarray*}
It is easy to check that the maximin score of $c$ in $V^*$ is at most $-t - 2(k-r)$, witnessed by  $D_{V^*}(c,w_1)$, and the maximin score of $x$ is at least $-t - 2(k-r)$, witnessed by  $D_{V^*}(x,w_x)$. Therefore, $c$ is no longer the unique winner in $V^*$ with respect to the maximin voting rule, and we have a contradiction.
\end{proof}

We are now ready to describe a set cover of size at most $k$ for $\UU$ based on $V^*$. Define $J \subseteq [t]$ as being the set of all indices $i$ for which the extension of $\lambda_i$ in $V^*$ belongs to $P_C^*$. Consider: \[\HH := \{S_i ~|~ i \in J\}.\] The set $\HH$ is our proposed set cover. Clearly, $|\HH| \leq k$. It remains to be shown that $\HH$ is a set cover. 

Assume, for the sake of contradiction, that there is an element $u_i \in \UU$ that is not covered by $\HH$. This means that for all $i \in J$, $u_i \in T_i$, and in the corresponding extensions of $\lambda_i$ in $V^*$, implying that $w_i \succ u_i$.  Further, for all $i \notin J$, we have that the extension of $\lambda_i$ in $V^*$ is consistent with: 
\[w_1 \succ \cdots \succ w_i \succ \cdots \succ w_x \succ x \succ S_i \succ c \succ T_i,\]
implying again that $w_i \succ u_i$ in these votes. Therefore, we have:
\[D_{V^*}(w_i,u_i) = D_V(w_i,u_i) + k + (t-k) = -2t + t = -t.\]
We know that the maximin score of $c$ is less than or equal to $-t$, since $D_{V^*}(c,l_1) = -t$, and we now have that the maximin score of $w_i$ is $-t$. This overrules $c$ as the unique winner in $V^*$, contradicting our assumption to that effect. This completes the proof. 
\qed\end{proof}

\subsection{Copeland Voting Rule}

We now describe the reduction for  \textsc{Possible Winner} parameterized by the number of candidates, for the Copeland voting rule.

\begin{theorem}
\label{thm:npk_copeland}
The \textsc{Possible Winner} problem for the Copeland voting rule, when parameterized by the number of candidates, does not admit a polynomial kernel unless~\caveat{}.
\end{theorem}

\begin{proof}
Let $(\UU,\FF,k)$ be an instance of \textsc{Small Universe Set Cover}, where $\UU = \{u_1,\ldots,u_m\}$, and $\FF = \{S_1, \ldots, S_t\}$. For the purpose of this proof, we assume (without loss of generality) that $m \geq 6$. We now construct an instance $(\CC,V,c)$ of \textsc{Possible Winner} as follows. 

\paragraph*{Candidates.}~$\CC := \UU \uplus \{z, c, d, w\}$.

\paragraph*{Partial Votes, $P$.}~The first part of the voting profile comprises of $m$ partial votes, and will be denoted by $P$. For each $1 \leq i \leq m$, we first consider a profile built on a total order:
\[\eta_i := \UU \setminus S_i \succ z \succ c \succ d \succ S_i \succ w\]
Now, we obtain a partial order $\lambda_i$ based on $\eta_i$ as follows:
\[\lambda_i := \eta_i \setminus \left( \{z,c\} \times \left( S_i \uplus \{d,w\} \right)\right)\]
The profile $P$ consists of $\{ \lambda_i ~|~ 1 \leq i \leq m \}$.

\paragraph*{Complete Votes, $Q$.}~We now describe the remaining votes in the profile, which are linear orders designed to achieve specific pairwise difference scores among the candidates. This profile, denoted by $Q$, is defined according to~lemma\nobreakspace \ref {thm:mcgarvey} to contain votes such that the pairwise score differences of $P \cup Q$ satisfy the following.

\begin{itemize}
\item $D(c,d) = t - 2k + 1$
\item $D(z,w) = t - 2k - 1$
\item $D(c,u_i) = t - 1$
\item $D(c,z) = t+1$ 
\item $D(c,w) = -t-1$ 
\item $D(u_i,d) = D(z,u_i) = t+1$, for all $1 \leq i \leq m$ 
\item $D(z,d) = t+1$
\item $D(u_i,u_j) = t+1$, for all $j \in [i+1,i+\lfloor m/2 \rfloor]$, where addition is modulo $m$.
\end{itemize}

We note that the for all $c,c' \in \CC$, the difference $|D(c,c') - D_P(c,c')|$ is always even, as long $t$ is odd and the number of sets in $\FF$ that contain an element $a \in \UU$ is always odd. Note that the latter can always be ensured without loss of generality: indeed, if $a \in \UU$ occurs in an even number of sets, then we can always add the set $\{a\}$ if it is missing and remove it if it is present, flipping the parity in the process. The number of sets $t$ can be assumed to be odd by adding a dummy element and a dummy set that contains the said element. It is easy to see that these modifications always preserve the instance. 

Thus the constructed instance of \textsc{Possible Winner} is $(\CC,V,c)$, where $V := P \cup Q$. We now turn to the proof of correctness. 

In the forward direction, let $\HH \subseteq \FF$ be a set cover of size at most $k$. Without loss of generality, let $|\HH| = k$ (since a smaller set cover can always be extended artificially) and let $\HH = \{S_1, \ldots, S_k\}$ (by renaming).  

If $i \leq k$, let: 
\[\lambda_i^* := \UU \setminus S_i \succ z \succ c \succ d \succ S_i \succ w\]
If $i > k$, let: 
\[\lambda_i^* := \UU \setminus S_i \succ d \succ S_i \succ w \succ z \succ c\]
Clearly $\lambda_i^*$ extends $\lambda_i$ for all $i \in 1 \leq i \leq k$. Let $V^*$ denote the extended profile consisting of the votes $\{ \lambda_i^* ~|~ 1 \leq i \leq k\} \cup Q$. We now claim that $c$ is the unique winner with respect to the Copeland voting rule in $V^*$. 

First, consider the candidate $z$. Between $z$ and $u_i$, even if $z$ loses to $u_i$ in every $\lambda_i^*$, because $D(z,u_i) = t+1$, $z$ wins the pairwise election between $z$ and $u_i$. The same argument holds between $z$ and $d$. Therefore, the Copeland score of $z$, no matter how the partial votes were extended, is at least $(m+1)$. 

Further, note that all other candidates (apart from $c$) have a Copeland score of less than $(m+1)$, because they are guaranteed to lose to at least three candidates (assuming $m \geq 6$). In particular, observe that $u_i$ loses to at least $\lfloor m/2 \rfloor$ candidates, and $d$ loses to all $u_i$ (merely by its position in the extended votes), and $w$ loses to all $u_i$ (because of way the scores were designed). Therefore, the Copeland score of all candidates in $\CC\setminus \{z,c\}$ is strictly less than the Copeland score of $z$, and therefore they cannot be possible (co-)winners.

Now we restrict our attention to the contest between $z$ and $c$. First note that $c$ beats every $u_i$: since the sets of $\HH$ form a set cover, every $u_i$ is present in some $\lambda^*_i$ for $i \leq k$, in a position after $c$. Since the difference score between $c$ and $u_i$ was $(t-1)$, even if $c$ suffered defeat in every other extension, we have the pairwise score of $c$ and $u_i$ being at least $t-1 - (t-1) + 1 = 1$, which implies that $c$ defeats every $u_i$ in their pairwise election. Note that $c$ also defeats $d$ by getting ahead of $d$ in $k$ votes, making its final score $t - 2k + 1 + k - (t-k) = 1$. Finally, $c$ is defeated by $w$, simply by the preset difference score. Therefore, the Copeland score of $c$ is $(m+2)$.

Now, all that remains to be done is to rule $z$ out of the running. Note that $z$ is defeated by $w$ in their pairwise election: this is because $z$ defeats $w$ in $k$ of the extended votes, and is defeated by $w$ in the remaining. This implies that its final pairwise score with respect to $w$ is at most $t - 2k - 1 + k - (t-k) = -1$. Also note that $z$ loses to $c$ because of its predefined difference score. Thus, the Copeland score of $z$ in the extended vote is exactly $(m+1)$, and thus $c$ is the unique winner of the extended vote. 

We now turn to the reverse direction. Let $P^*$ be an extension of $P$ such that $V^* := P^* \cup Q$ admits $c$ as an unique winner with respect to the Copeland voting rule. As with the argument for Maximin voting rule, we first argue that $P^*$ must admit a certain structure, which will lead us to an almost self-evident set cover for $\UU$.

Let us denote by $P^*_C$ the set of votes in $P^*$ which are consistent with $c \succ d$, and let $P^*_W$ be the set of votes in $P^*$ which are consistent with $w \succ z$. Note that the votes in $P^*_C$ necessarily have the form:
\[\lambda_i^* := \UU \setminus S_i \succ z \succ c \succ d \succ S_i \succ w\]
and those in $P^*_W$ have the form: 
\[\lambda_i^* := \UU \setminus S_i \succ d \succ S_i \succ w \succ z \succ c\]
It is easy to check that this structure is directly imposed by the relative orderings that are fixed by the partial orders.

Before we argue the details of the scores, let us recall that in any extension of $P$, $z$ loses to $c$ and $z$ wins over $d$ and all candidates in $\UU$. Thus the Copeland score of $z$ is at least $(m+1)$. On the other hand, in any extension of $P$, $c$ loses to $w$, and therefore the Copeland score of $c$ is at most $(m+2)$. (These facts follow from the analysis in the forward direction.) 

Thus, we have the following situation. If $z$ wins over $w$, then $c$ cannot be the unique winner in the extended vote, because the score of $z$ goes up to $(m+2)$. Similarly, $c$ cannot afford to lose to any of $\UU \cup \{d\}$, because that will cause its score to drop below $(m+2)$, resulting in either a tie with $z$, or worse. These facts will successively lead us to the correctness of the reverse direction. 

Now, let us return to the sets $P^*_C$ and $P^*_W$. If $P^*_C$ has more than $k$ votes, then $z$ wins over $w$: the final score of $z$ is at least $t - 2k - 1 + (k+ 1) - (t-k-1) = 1$, and we have a contradiction. If $P^*_C$ has fewer than $k$ votes, then $c$ loses to $d$, with a score of at most $t - 2k + 1 + (k-1) - (t-k+1) = -1$, and we have a contradiction. 

Finally, suppose the sets corresponding to the votes of $P^*_C$ do not form a set cover. Consider an element $u_i$ of $\UU$ not covered by the union of these sets. Observe that $c$ now loses the pairwise election between itself and $u_i$ and is no longer in the running for being the unique winner in the extended vote. Therefore, the sets corresponding to the votes of $P^*_C$ form a set cover of size exactly $k$, as desired.
\qed\end{proof}

\subsection{Bucklin Voting Rule}

We now describe the reduction for  \textsc{Possible Winner} parameterized by the number of candidates, for the Bucklin voting rule.

\begin{theorem}
\label{thm:npk_bucklin}
The \textsc{Possible Winner} problem for the Bucklin voting rule, when parameterized by the number 
of candidates, does not admit a polynomial kernel unless~\caveat{}.
\end{theorem}

\begin{proof}
Let $(\UU,\FF,k)$ be an instance of \textsc{Small Universe Set Cover}, where $\UU = \{u_1,\ldots,u_m\}$, and 
$\FF = \{S_1, \ldots, S_t\}$. Without loss of generality, we assume that $t>k+1$, and that every set has at 
least two elements. We now construct an instance $(\CC,V,c)$ of \textsc{Possible Winner} as follows. 

\paragraph*{Candidates.}~$\CC := \UU \uplus \{z, c, a\} \uplus W \uplus D_1 \uplus D_2 \uplus D_3$, where $D_1$, $D_2$, and $D_3$ are any sets such that $|D_1|=m$, $|D_2|=2m$, and $|D_3|= $. $W:=\{w_1, w_2, \dots, w_{2m} \}$.

\paragraph*{Partial Votes, $P$.}~The first part of the voting profile comprises of $t$ partial votes, and will be denoted by $P$. For each $1 \leq i \leq t$, we first consider a profile built on a total order:
\begin{equation*}
\eta_i := \UU \setminus S_i \succ S_i \succ w_{i\% m} \succ w_{(i+1)\% m} \succ z \succ c \succ D_3 \succ \dots
\end{equation*}

Now, we obtain a partial order $\lambda_i$ based on $\eta_i$ as follows:
\[\lambda_i := \eta_i \setminus \left(\left( \{w_{i\% m}, w_{(i+1)\% m}, z,c\} \uplus D_3\right) \times S_i \right)\]
The profile $P$ consists of $\{ \lambda_i ~|~ 1 \leq i \leq t \}$.

\paragraph*{Complete Votes, $Q$.} 

\begin{eqnarray*}
 t-k-1 &:& D_1 \succ z \succ c \succ \dots \\
 1 &:& D_1 \succ c \succ a \succ z \succ \dots \\
 k-1 &:& D_2 \succ \dots 
\end{eqnarray*}

We now show that $(\UU,\FF,k)$ is an yes instance if and only if $(\CC,V,c)$ is an yes instance. 
Suppose, $\{S_j : j\in J\}$ forms a set cover. Then, consider the following extension of $P$ : 

\begin{equation*}
 (\UU \setminus S_j) \succ w_{j\% m} \succ w_{(j+1)\% m} \succ z \succ c \succ D_3 \succ S_j \dots, \text{ for } j\in J
\end{equation*}
\begin{equation*}
 (\UU \setminus S_j) \succ S_j \succ w_{j\% m} \succ w_{(j+1)\% m} \succ z \succ c \succ D_3 \succ \dots, \text{ for } j\notin J
\end{equation*}

We claim that in this extension, $c$ is the unique winner with Bucklin score $(m+2)$. First, let us establish the score of $c$. The candidate $c$ is already within the top $(m+1)$ choices in $(t-k)$ of the complete votes. In all the sets that form the set cover, $c$ is ranked within the first $(m+2)$ votes in the proposed extension of the corresponding vote (recall that every set has at least two elements). Therefore, there are a total of $t$ votes where $c$ is ranked within the top $(m+2)$ preferences. Further, consider a candidate $v \in \UU$. Such a candidate is not within the top $(m+2)$ choices of any of the complete votes. Let $S_i$ be the set that covers the element $v$. Note that in the extension of the vote $\lambda_i$, $v$ is not ranked among the top $(m+2)$ spots, since there are at least $m$ candidates from $D_3$ getting in the way. Therefore, $v$ has strictly fewer than $t$ votes where it is ranked among the top $(m+2)$ spots, and thus has a higher Bucklin score than $c$.

Now, the candidate $z$ is within the top $(m+2)$ ranks of at most $(t-k-1)$ voters among the complete votes. In the votes corresponding to the sets \emph{not} in the set cover, $z$ is placed beyond the first $(m+2)$ spots. Therefore, the number of votes where $z$ is among the top $(m+2)$ candidates is at most $(t-1)$, which makes its Bucklin score strictly larger than $(m+2)$. 

The candidates from $W$ are within the top $(m+2)$ positions only in a constant number of votes. The candidates $D_1 \cup \{a\}$ have $(t-k)$ votes (among the complete ones) in which they are ranked among the top $(m+2)$ preferences, but in all extensions, these candidates have ranks below $(m+2)$. Finally, the candidates in $D_3$ do not feature in the top $(m+2)$ positions of any of the complete votes, and similarly, the candidates in $D_2$ do not feature in the top $(m+2)$ positions of any of the extended votes. Therefore, the Bucklin scores of all these candidates is easily seen to be strictly larger than $(m+2)$, concluding the argument in the forward direction. 

Now, consider the reverse direction. Suppose, $(\CC,V,c)$ is an yes instance. For the same reasons described in the forward direction, observe that only the following candidates can win depending upon how the partial 
preferences get extended - either one of the candidates in $\UU$, or one of $z$ or $c$. Note that the Bucklin score of $z$ in any extension is at most $(m+3)$. Therefore, the Bucklin score of $c$ has to be $(m+2)$ or less. Among the complete votes $Q$, there are $(t-k)$ votes where the candidate $c$ appears in the top $(m+2)$ positions. To get majority within top $(m+2)$ positions, $c$ should be within top $(m+2)$ positions for at least $k$ of the extended votes in $P$.
Let us call these set of votes $P^{\prime}$. Now notice that whenever $c$ comes within top 
$(m+2)$ positions in a valid extension of $P$, the candidate $z$ also comes within top $(m+2)$ positions in the same vote. However, the candidate $z$ is already ranked among the top $(m+2)$ candidates in $(t-k-1)$ complete votes. Therefore, $z$ can appear within top $(m+2)$ positions in \emph{at most} $k$ extensions (since $c$ is the unique winner), implying that $|P^{\prime}|=k$. Further, note that the Bucklin score of $c$ cannot be strictly smaller than $(m+2)$ in any extension. Indeed, the candidate $c$ features in only one of the complete votes within the top $(m+1)$ positions, and it would have to be within the top $(m+1)$ positions in at least $(t-1)$ extensions. However, as discussed earlier, this would give $z$ exactly the same mileage, and therefore its score Bucklin score would be $(m-1)$ or even less; contradicting our assumption that $c$ is the unique winner. 

Now, we claim that the $S_i$'s corresponding to the votes in $P^{\prime}$ form a set cover for $\UU$. If not,l there is an element $x\in \UU$ that is uncovered. Observe that $x$ appears within top $m$ positions in all the extensions of the votes in $P^\prime$, by assumption. Further, in all the remaining extensions, since $z$ is not present among the top $(m+2)$ positions, we only have room for two candidates from $W$. The remaining positions must be filled by all the candidates corresponding to elements of $\UU$. Therefore, $x$ appears within the top $(m+2)$ positions of all the extended votes. Since these constitute half the total number of votes, we have that $x$ ties with $c$ in this situation, a contradiction. 
\qed\end{proof}

\subsection{Ranked Pairs Voting Rule}

We now describe the reduction for  \textsc{Possible Winner} parameterized by the number of candidates, for the ranked pairs voting rule.

\begin{theorem}
\label{thm:npk_rankedpairs}
The \textsc{Possible Winner} problem for the ranked pairs voting rule, when parameterized by the 
number of candidates, does not admit a polynomial kernel unless~\caveat{}.
\end{theorem}

\begin{proof}
Let $(\UU,\FF,k)$ be an instance of \textsc{Small Universe Set Cover}, where $\UU = \{u_1,\ldots,u_m\}$, and 
$\FF = \{S_1, \ldots, S_t\}$. Without loss of generality, we assume that $t$ is even. We now construct an 
instance $(\CC,V,c)$ of \textsc{Possible Winner} as follows. 

\paragraph*{Candidates.}~$\CC := \UU \uplus \{a, b, c, w \}$.

\paragraph*{Partial Votes, $P$.}~The first part of the voting profile comprises of $t$ partial votes, and will 
be denoted by $P$. For each $1 \leq i \leq t$, we first consider a profile built on a total order:
\[\eta_i := \UU \setminus S_i \succ S_i \succ b \succ a \succ c \succ \dots\]
Now, we obtain a partial order $\lambda_i$ based on $\eta_i$ as follows:
\[\lambda_i := \eta_i \setminus \left(\{ a, c\} \times \left( S_i \uplus \{ b \} \right)\right)\]
The profile $P$ consists of $\{ \lambda_i ~|~ 1 \leq i \leq t \}$.

\paragraph*{Complete Votes, $Q$.} We add complete votes such that along with the already determined 
pairs from the partial votes $P$, we have following.

\begin{itemize}
\item $D(v_i, c) = 2, \forall 1\le i\le t $
\item $D(c,b) = 4t$
\item $D(c,w) = t+2$
\item $D(b,a) = 2k + 4$ 
\item $D(w,a) = 4t$ 
\item $D(a,c) = t+2$, $\forall$ $1 \leq i \leq m$ 
\item $D(w, v_i) = 4t$
\end{itemize}

We now show that $(\UU,\FF,k)$ is an yes instance if and only if $(\CC,V,c)$ is an yes instance. 
Suppose, $\{S_j : j\in J\}$ forms a set cover. Then, consider the following extension of $P$ : 
\[ \UU \setminus S_j \succ a \succ c \succ S_j \succ b \succ \dots , \forall j\in J\]
\[ \UU \setminus S_j \succ S_j \succ b \succ a \succ c \succ \dots , \forall j\notin J\]
We claim that the candidate $c$ is the unique winner in this extension. Note that the pairs $(w \succ a)$ and  $(w \succ v_i)$ (for all $1 \leq i \leq t$) get locked first (since these differences are clearly the highest and unchanged). The pair $(c,b)$ gets locked next, with a difference score of $(3t + 2k)$. Now, since the votes in which $c \succ b$ are based on a set cover of size at most $k$, the the pairwise difference between $b$ and $a$ becomes at least $2k + 4 - k + (t - k) = t + 4$. Therefore, the next pair to get locked is $b \succ a$. Finally, for any element $v_i \in \UU$, the difference $D(v_i,c)$ is at most $2 + (t-1) = t+1$, since there is at least one vote where $c \succ v_i$ (given that we used a set cover in the extension). It is now easy to see that the next highest pairwise difference is between $c$ and $w$, so the ordering $c \succ w$ gets locked, and at this point, by transitivity, $c$ is superior to $w, b, a$ and all $v_i$. It follows that $c$ wins the election irrespective the 
sequence in which pairs are considered subsequently.      

Now suppose, $(\CC,V,c)$ is an yes instance. Notice that, irrespective of the extension of the 
votes in $P$, $c \succ b, w \succ a, w \succ v_i ,\forall 1\le i\le t, $ are locked first. Now, if $b \succ c$ in all the extended votes, then it is easy to check that $b \succ a$ gets locked next, with a difference score of $2k+4+t$; leaving us with $D(v_i,c) = t+2 = D(c,w)$, where $v_i \succ c$ could be a potential lock-in. This implies the possibility of a $v_i$ being a winner in some choice of tie-breaking, a contradiction to the assumption that $c$ is the unique winner. Therefore, there are at least some votes in the extended profile where $c \succ b$. We now claim that there are at most $k$ such votes. Indeed, if there are more, then $D(b,a) = 2k + 4 - (k+1) + (t-k-1) = t + 2$. Therefore, after the forced lock-ins above, we have $D(b,a) = D(c,w) = D(a,c) = t + 2$. Here, again, it is possible for $a \succ c$ to be locked in before the other choices, and we again have a contradiction.  

Finally, we have that $c \succ b$ in at most $k$ many extensions in $P$. Call the set of indices of these extensions $J$. We claim that $\{ S_j : j\in J \}$ forms a set cover. If not, then suppose an element $v_i$ is not covered by $\{ S_j : j\in J \}$. Then, the candidate $v_i$ comes before $c$ in all the extensions which makes 
$ D(v_i,c) $ become $(t+2)$, which in turn ties with $D(c,w)$. This again contradicts the fact that $c$ is the unique winner. Therefore, if there is an extension that makes $c$ the unique winner, then we have the desired set cover. 
\qed\end{proof}

\section{Polynomial Kernels for Coalitional Manipulation Problem}

We now describe a kernelization algorithm for any scoring rule which has the following properties. For $m\in \mathbb{N}$, let $(\alpha_1, \ldots, \alpha_m)$ and $(\alpha_1^{\prime}, \ldots, \alpha_{m+1}^{\prime})$ be the normalized score vectors for a scoring rule for an election with $m$ and $(m+1)$ candidates respectively. Then $\alpha_1 = O(poly(m))$ and $\alpha_i = \alpha_i^{\prime}$ for all $1\le i\le m$. 
For ease of exposition, we present the result for the Borda voting rule only.
\setcounter{reductionrule}{0}

\begin{theorem}
\label{thm:pk_borda} 
The \textsc{Coalitional Manipulation} problem for the Borda voting rule admits a polynomial kernel 
when the number of manipulators is $O(poly(m))$.
\end{theorem}

\begin{proof}
 Let $c$ be the candidate whom the manipulators aim to make winner. Let $M$ be the set of manipulators and $\mathcal{C}$ candidate set. Let $s_{NM}(x)$ be the score of candidate $x$ from the votes of the non-manipulators. Without loss of generality, we assume that, all the manipulators place $c$ at top position in their votes. Hence, the final score of $c$ is $= s_{NM}(c) + |M|(m-1)$, which we denote by $s(c)$.
 Now, if $s_{NM}(x) \ge s(c)$ for any $x \ne c$, then $c$ cannot win and we output \textit{no}. Hence, we assume that $s_{NM}(x) < s(c)$ for all $x \ne c$. Now let us define $s_{NM}^*(x)$ as follows.
 $$ s_{NM}^*(x) := \max \{ s_{NM}(x), s_{NM}(c) \} $$
 Also define $s_{NM}^*(c)$ as follows.
 $$ s_{NM}^*(c) := s_{NM}(c) - |M|$$
 We define a \textsc{Coalitional Manipulation} instance with $(m+1)$ candidates as $(\mathcal{C}^{\prime},NM,M,c)$, where $\mathcal{C}^{\prime} = \mathcal{C} \uplus \{d\}$ is the candidate set, $M$ is the set of manipulators, $c$ is the distinguished candidates, and $NM$ is the non-manipulators' vote is such that it generates score of $x\in \mathcal{C}$ to be $ K + (s_{NM}^*(x) - s_{NM}(c))$, where $K \in \mathbb{N}$ is same for $x\in \mathcal{C}$, and the score of $d$ is less than $ K - m|M| $. The existence of such a $NM$ of size $poly(m)$ is due to~Lemma\nobreakspace \ref {score_gen}. Hence,
 once we show the equivalence of these two instances, we have a kernel whose size is polynomial in $m$. 
 The equivalence of the two instances follows from the facts that (1) the new instance has $(m+1)$ candidates and $c$ is always placed at the top position without loss of generality, $c$ gets $|M|$ score more than the initial instance and this is compensated in $s_{NM}^*(c)$, (2) the relative score difference from the final score of $c$ to the current score of every $x\in \mathcal{C}\setminus \{c\}$ is same in both the instances, and (3) in the new instance, we can assume without loss of generality that the candidate $d$ will be placed in the second position in all the manipulators' votes.
\qed\end{proof}

We now move on to the voting rules that are based on weighted majority graph.
The reduction rules modify the weighted majority graph maintaining the property that there exists a set of votes that can realize the modified weighted majority graph. In particular, the final weighted majority graph is realizable with a set of votes.

\setcounter{reductionrule}{0}
\begin{theorem}
\label{thm:pk_maximin} 
The \textsc{Coalitional Manipulation} problem for the maximin voting rule admits a polynomial kernel 
when the number of manipulators is $O(poly(m))$.
\end{theorem}

\begin{proof}
 Let $c$ be the distinguished candidate of the manipulators. Let $M$ be the set of manipulators. 
 We can assume that $ |M| \ge 2$ since for $|M| = 1$, the problem is in $\Pb$ \citep{bartholdi1989computational}.
 Define $s$ to be $\min_{x\in C\setminus \{c\}} D_{(\VV\setminus M)}(c,x)$. 
 So, $s$ is the maximin score of the candidate $c$ from the votes except from $M$. Since, the maximin voting 
 rule is monotone, we can assume that the voters in $M$ put the candidate $c$ at top position of their 
 preferences. Hence, $c$'s final maximin score will be $s+|M|$. This provides the following reduction 
 rule.
 
 \begin{reductionrule}\label{rr:condorcet_winner}
  If $ s+|M| \ge 0 $, then output \YES{}.
 \end{reductionrule}
 
 In the following, we will assume $s+|M|$ is negative. Now we propose the following reduction rules on the weighted majority graph.
 
 \begin{reductionrule}\label{rr:weight_up}
  If $D_{(\VV\setminus M)}(c_i, c_j) < 0$ and $D_{(\VV\setminus M)}(c_i, c_j) > 2|M| + s$, then make $D_{(\VV\setminus M)}(c_i, c_j)$ either $2|M|+s+1$ or $2|M|+s+2$ whichever keeps the parity of 
  $D_{(\VV\setminus M)}(c_i, c_j)$ unchanged.
 \end{reductionrule}
 
 If $D_{(\VV\setminus M)}(c_i, c_j) > 2|M| + s,$,then $D_{\VV}(c_i, c_j) > |M| + s$ irrespective of the way the manipulators vote. Hence, whether or not the maximin score of $c_i$ and $c_j$ will exceed the maximin score of $c$ does not 
 gets affected by this reduction rule. Hence, ~reduction rule\nobreakspace \ref {rr:weight_up} is sound.
 
 \begin{reductionrule}\label{rr:weight_lw}
  If $D_{(\VV\setminus M)}(c_i, c_j) < s $, then make $D_{(\VV\setminus M)}(c_i, c_j)$ either $s-1$ or
  $s-2$ whichever keeps the parity of $D_{(\VV\setminus M)}(c_i, c_j)$ unchanged.
 \end{reductionrule}
 
 The argument for the correctness of reduction rule\nobreakspace \ref {rr:weight_lw} is similar to the argument for reduction rule\nobreakspace \ref {rr:weight_up}. 
 Here onward, we may assume that whenever $D_{(\VV\setminus M)}(c_i, c_j) < 0$, 
 $ s-2 \le D_{(\VV\setminus M)}(c_i, c_j) \le 2|M|+s+2 $
 
 \begin{reductionrule}\label{rr:weight_closer}
  If $ s < -4|M| $ then subtract $ s+5|M| $ whenever $ s-2 \le D_{(\VV\setminus M)}(c_i, c_j) \le 2|M|+s+2 $.
 \end{reductionrule}
 
 The correctness of reduction rule\nobreakspace \ref {rr:weight_closer} follows from the fact that it adds linear fixed offsets to all the 
 edges of the weighted majority graph. Hence, if there a voting profile of the voters in $M$ that makes the candidate 
 $c$ win in the original instance, the same voting profile will make $c$ win the election in the reduced 
 instance and vice versa.
 
 Now, we have an weighted majority graph with $O(|M|)$ weights for every edge. Also, all the weights have uniform parity 
 and thus the result follows from lemma\nobreakspace \ref {thm:mcgarvey}.
\qed\end{proof}
\setcounter{reductionrule}{0}

The proof of the following theorem is analogous to the proof of~Theorem\nobreakspace \ref {thm:pk_maximin}. In the interest of space, we omit the proof.
\begin{theorem}
\label{thm:pk_copeland} 
The \textsc{Coalitional Manipulation} problem for the Copeland voting rule admits a polynomial kernel 
when the number of manipulators is $O(poly(m))$.
\end{theorem}

\begin{proof}
We apply the following reduction rule.
 \begin{reductionrule}\label{rr:weight_up}
  If $D_{(\VV\setminus M)}(c_i, c_j) > |M| $, then make $D_{(\VV\setminus M)}(c_i, c_j)$ 
  either $|M|+1$ or $|M|+2$ whichever keeps the parity of 
  $D_{(\VV\setminus M)}(c_i, c_j)$ unchanged.
 \end{reductionrule}
 Given any votes of $M$, $D_{\VV}(c_i, c_j) > 0 $ in the original instance if and only if 
 $D_{\VV}(c_i, c_j) > 0 $ in the reduced instance. Hence each candidate has the same Copeland score 
 and thus the reduction rule is correct.
 
 Now, we have an weighted majority graph with $O(|M|)$ weights for every edges. Also, all the weights have uniform parity. 
 From lemma\nobreakspace \ref {thm:mcgarvey}, we can realize the weighted majority graph using $O(m^2.|M|)$ votes. 
\qed\end{proof}
\setcounter{reductionrule}{0}

\begin{theorem}
\label{thm:pk_rankedpair} 
The \textsc{Coalitional Manipulation} problem for the ranked pairs voting rule admits a polynomial kernel 
when the number of manipulators is $O(poly(m))$.
\end{theorem}

\begin{proof}
 Consider all non-negative $D_{(\VV\setminus M)}(c_i, c_j)$ and arrange them in non-decreasing order. Let 
 the ordering be $x_1, x_2, \dots, x_l$ where $ l = {m \choose 2} $. Now keep applying following reduction 
 rule till possible. Define $x_0 = 0$.
 \begin{reductionrule}
  If there exist any $i$ such that, $x_i - x_{i-1} > |M|+2$, subtract an even offset to all 
  $x_i, x_{i+1}, \dots, x_l$ such that $x_i$ becomes either $(x_{i-1} + |M| + 1)$ or 
  $(x_{i-1} + |M| + 2)$.
 \end{reductionrule} 
 The reduction rule is correct since for any set of votes by $M$, for any four candidates $a, b, x, y \in \CC$, 
 $ D(a,b) > D(x,y) $ in the original instance if and only if $ D(a,b) > D(x,y) $ in the reduced instance. Now, 
 we have an weighted majority graph with $O(m^2.|M|)$ weights for every edges. Also, all the weights have uniform 
 parity and hence can be realized with $O(m^4.|M|)$ votes~Lemma\nobreakspace \ref {thm:mcgarvey}. 
\qed\end{proof}

\section{Conclusion and Future Work}
Here, we showed that the \textsc{Possible Winner} problem does not admit a polynomial 
kernel for many common voting rules under the complexity theoretic assumption that 
\caveat is not true. We also showed the existence of polynomial kernels for the 
\textsc{Coalitional Manipulation} problem for many common voting rules. This shows that the \textsc{Possible Winner} 
problem is a significantly harder problem than the \textsc{Coalitional Manipulation} problem, 
although both the problems are in $\NPCb{}$.

There are other interesting parameterizations 
of these problems for which fixed parameter tractable algorithms are known but 
the corresponding kernelization questions are still open. One such parameter is the 
total number of pairs $s$ in all the votes for which an ordering has not been specified. 
With this parameter, a simple $O(2^s. \text{ poly($m,n$)})$ algorithm is known \citep{betzler2009multivariate}. 
However, the corresponding kernelization question is still open.


\bibliographystyle{apalike}
\bibliography{pw}

\end{document}